\documentclass[
final
,nomarks
]{dmtcs-episciences}

\usepackage[utf8]{inputenc}
\usepackage{subfigure}

\usepackage[T1]{fontenc}
\usepackage{grffile}
\usepackage{longtable}
\usepackage{wrapfig}
\usepackage{rotating}
\usepackage[normalem]{ulem}
\usepackage{amsmath}
\usepackage{textcomp}
\usepackage{amssymb,amsthm}
\usepackage{capt-of}
\usepackage{hyperref}

\newtheorem{theorem}{Theorem}
\newtheorem{proposition}{Proposition}
\newtheorem{corollary}{Corollary}
\newtheorem{claim}{Claim}

\newcommand\cuparrow{%
  \mathrel{\ooalign{\hss$\cup$\hss\cr%
  \kern0.3ex\raise0.7ex\hbox{\scalebox{0.7}{$\downarrow$}}}}}

\newcommand\bigcuparrow{%
  \mathrel{\ooalign{\hss$\bigcup$\hss\cr%
  \kern0.55ex\raise0.7ex\hbox{\scalebox{0.7}{$\downarrow$}}}}}

\newcommand{\tw}{{\mathbf{tw}}}
\newcommand{\ww}{{\mathbf{w}}}
\newcommand{\Ocal}{{\mathcal{O}}}
\newcommand{\Xcal}{{\mathcal{X}}}
\newcommand{\Gcal}{{\mathcal{G}}}
\newcommand{\Ecal}{{\mathcal{E}}}
\newcommand{\Acal}{{\mathcal{A}}}
\newcommand{\aaa}{{a}}
\newcommand{\bbb}{{b}}
\newcommand{\dptable}{{\bf r}}
\newcommand{\hs}{\widehat{S}}

\newcommand{\sbf}{\textbf{s}}
\newcommand{\rbf}{{\bf r}}
\newcommand{\degs}[2]{{\texttt{deg}}_{#1}(#2)}
\newcommand{\ef}{\varnothing}

\newcommand{\rmc}{\textsf{rmc}}
\newcommand{\ins}{\textsf{ins}}
\newcommand{\shift}{\textsf{shft}}
\newcommand{\glue}{\textsf{glue}}
\newcommand{\glues}{\textsf{glueset}}
\newcommand{\proj}{\textsf{proj}}
\newcommand{\join}{\textsf{join}}
\newcommand{\opt}{\textsf{opt}}
\newcommand{\reduce}{\textsf{reduce}}

\newcommand{\probl}[3]{
\begin{flushleft}
\fbox{
\begin{minipage}{14.45cm}
\noindent {\sc #1}\\
          {\bf Input:} #2\\
          {\bf Output:} #3
\end{minipage}}
\medskip
\end{flushleft}
}

\usepackage[round]{natbib}

\author[Julien Baste]{Julien Baste}
\title[\textsc{Leafed Induced Subtree} in restricted graph classes]{The \textsc{Leafed Induced Subtree} in chordal and bounded treewidth graphs}
\affiliation{Univ. Lille, CNRS, Centrale Lille, UMR 9189 CRIStAL, F-59000 Lille, France}

\begin{document}

\publicationdata{vol. 28:2}{2026}{9}{10.46298/dmtcs.11023}{2023-03-03;
2023-03-03; 2024-07-29; 2024-12-02; 2026-01-09}{2026-01-11}
\maketitle
\begin{abstract}

  In the \textsc{Fully Leafed Induced Subtrees}, one is given a graph $G$ and two integers $a$ and $b$ and the question is to find an induced subtree of $G$ with $a$ vertices and at least $b$ leaves. This problem is known to be \textsf{NP}-complete even when the input graph is $4$-regular.
  Polynomial algorithms are known when the input graph is restricted to be a tree or series-parallel.
  In this paper we generalize these results by providing an \textsf{FPT} algorithm parameterized by treewidth.
  We also provide a polynomial algorithm when the input graph is restricted to be a chordal graph.

\end{abstract}

\section{Introduction}

An important focus in the graph community is to find substructures in a graph.
Since the 80's, a particular attention is given to subgraphs that are trees, see for instance the works of \cite{ErSaSo1986}, \cite{KlWe1991}, \cite{PaTcXu1984}, and \cite{WaArUn2014}.
While finding a spanning tree in a graph is almost trivial, adding constraints on the required spanning tree quickly increases the complexity of the problem.
For instance, one can ask for a spanning tree that maximizes its number of leaves.
This problem, called \textsc{Maximum Leaf Spanning Tree}, is known to be \textsf{NP}-hard even when the input graph is $4$-regular as shown by~\cite{GaJo1979}.

Finding subgraphs of a graph that are trees with additional constraints find application in telecommunication networks, as exposed by~\cite{BoChLi2005} and \cite{ChLjRa2015}, in data mining, as presented by~\cite{DeFeTiSaMc2014}, or in information retrieval, as shown by \cite{Za2005}.
In this paper, we focus on \textsc{Fully Leafed Induced Subtrees}.
This notion is introduced by~\cite{BlCaGo2018} and is later used by~\cite{AbBlGo2018} and \cite{BlCaGoLaNaVa2018}.
In this problem, one is given a graph and the goal is to find the induced subtree of this graph that maximizes the number of leaves.
As discussed by~\cite{BlCaGoLaNaVa2018}, the fully leafed induced subtrees are expected to have application in chemical graph theory where they can represent natural structures like molecular network. 

For full generality, we considered the {\textsc{Restricted Leafed Induced Subtree}} problem.
\probl
{\textsc{Restricted Leafed Induced Subtree}}
{A graph $G$, a vertex $u\in V(G)$, and two integers $\aaa$ and $\bbb$.}
{Does there exist an induced subtree of $G$ containing $u$ as an internal vertex with exactly $\aaa$ vertices and at least $\bbb$ leaves?}

First note that when $\aaa \leq \bbb$, the problem is trivially solvable.
Through the paper, we always assume that $\aaa > \bbb$.
\cite{BlCaGoLaNaVa2018} show that the problem is \textsf{NP}-hard.
They also  provide a polynomial algorithm to solve {\textsc{Restricted Leafed Induced Subtree}} in trees and ask whether the problem remains polynomial in chordal graphs or in graphs of bounded treewidth.
\cite{AbBlGo2018} provide a first step in the resolution of this problem by providing a polynomial algorithm for series-parallel graphs, i.e., graphs of treewidth at most $2$.

The remainder of this paper provides positive answers to the open questions asked by~\cite{BlCaGoLaNaVa2018} by providing a polynomial algorithm for  {\textsc{Restricted Leafed Induced Subtree}} when the input graph is chordal as well as an \textsf{FPT} algorithm parameterized by treewidth.
In Section~\ref{sec:prelim}, we provide some needed definitions.
In Section~\ref{sec:simple}, we provide a simple dynamic programming algorithm that solves {\textsc{Restricted Leafed Induced Subtree}} in polynomial time when the input graph is chordal. In Section~\ref{sec:complexe} we present
 an \textsf{FPT} algorithm parameterized by treewidth based on the techniques introduced by~\cite{BoCyKrNe2015}.

\section{Preliminaries}
\label{sec:prelim}
\noindent\textbf{Basic definitions.}
Given two integers $a$ and $b$, we denote by $[a,b]$ the set of integers $\{x \mid a \leq x \leq b\}$.

In the paper, all the graphs are non-oriented and without loops nor multi-edges.
Given a graph $G$, we denote by $V(G)$ (resp. $E(G)$), the set of vertices (resp. edges) of $G$.
Given a vertex $v$ of $G$, we denote by $\texttt{deg}_{G}(v)$ the degree of the vertex $v$ in the graph $G$.
Given a set $S \subseteq V(G)$, we denote by $G[S]$ the subgraph of $G$ induced by $S$.

\bigskip
\noindent\textbf{Tree decompositions.} A \emph{tree decomposition} of a graph $G$ is a pair ${\cal D}=(T,{\cal X})$, where $T$ is a tree
and ${\cal X}=\{X_{t}\mid t\in V(T)\}$ is a collection of subsets of $V(G)$
 such that:
\begin{itemize}
\item $\bigcup_{t \in V(T)} X_t = V(G)$,
\item for every edge $\{u,v\} \in E$, there is a $t \in V(T)$ such that $\{u, v\} \subseteq X_t$, and
\item for each $\{x,y,z\} \subseteq V(T)$ such that $z$ lies on the unique path between $x$ and $y$ in $T$,  $X_x \cap X_y \subseteq X_z$.
\end{itemize}
We call the vertices of $T$ {\em nodes} of ${\cal D}$ and the sets in ${\cal X}$ {\em bags} of ${\cal D}$. The \emph{width} of a  tree decomposition ${\cal D}=(T,{\cal X})$ is $\max_{t \in V(T)} |X_t| - 1$.
The \emph{treewidth} of a graph $G$, denoted by $\tw(G)$, is the smallest integer $w$ such that there exists a tree decomposition of $G$ of width at most $w$.

\bigskip

\noindent
\textbf{Nice tree decompositions.} Let ${\cal D}=(T,{\cal X})$
be a tree decomposition of $G$, $r$ be a vertex of $T$, and   ${\cal G}=\{G_{t}\mid t\in V(T)\}$ be
a collection of subgraphs of   $G$, indexed by the vertices of $T$.
We say that the triple $({\cal D},r,{\cal G})$ is a
\emph{nice tree decomposition} of $G$ if the following conditions hold:
\begin{itemize}

\item $X_{r} = \emptyset$ and $G_{r}=G$.
\item Each node of ${\cal D}$ has at most two children in $T$.
\item For each leaf $t \in V(T)$, $X_t = \emptyset$ and $G_{t}=(\emptyset,\emptyset).$ Such $t$ is called a {\em leaf node}.
\item If $t \in V(T)$ has exactly one child $t'$, then either:
\begin{itemize}
\item $X_t = X_{t'}\cup \{v_{\rm insert}\}$ for some $v_{\rm insert} \not \in X_{t'}$ and $G_{t}=G[V(G_{t'})\cup\{v_{\rm insert}\}]$.
  In this case, the node $t$ is called \emph{introduce vertex}  node   and the vertex $v_{\rm insert}$ is the {\em insertion vertex} of $X_{t}$.
\item $X_t = X_{t'} \setminus \{v_{\rm forget}\}$ for some $v_{\rm forget} \in X_{t'}$ and $G_{t}=G_{t'}$.
  In this case, the node $t$ is called   \emph{forget vertex} node and $v_{\rm forget}$ is the {\em forget vertex} of $X_{t}$.
\end{itemize}
\item If $t \in V(T)$ has exactly two children $t'$ and $t''$, then $X_{t} = X_{t'} = X_{t''}$, and $G_t = G[V(G_{t'})\cup V(G_{t''})]$. The node $t$ is called a \emph{join} node.
\end{itemize}

As discussed by~\cite{Kl1994}, given a tree decomposition, it is possible to transform it in linear time to a {\sl nice} new one of the same width.
Note that for each $t \in T$, $G_t$ corresponds to the subgraph induced by all the vertices introduced in the nodes of the subtree of $T$ rooted at $t$.
Moreover, \cite{Ko2021} shows that we can find in time $2^{\Ocal(\tw)}\cdot n$ a tree decomposition of width $2\cdot \tw+1$ of any graph $G$. 
Since, in this section, we focus on single-exponential algorithms, we may assume that a nice tree decomposition of width $w = 2 \cdot \tw +1$
is given with the input.
For each $t \in V(T)$, we define $V_t = V(G_t)$.

\section{Chordal graphs}
\label{sec:simple}

In this section we present a polynomial time algorithm that solves \textsc{Restricted Leafed Induced Subtree} when the input graph is chordal.
Given a $n$-vertex chordal graph $G$, using the perfect elimination scheme of chordal graphs, as done by~\cite{Kl1994} in his Lemma 13.1.2, for $k$-trees, a nice tree decomposition $((T,\mathcal{X}), r, \mathcal{G})$ such that for each $t \in V(T)$, $X_t$ induces a clique in $G$, can be obtained in time $O(n^2)$.
Note that in this case, the size of the set $X_t$, $t \in V(T)$, is unbounded.
Our algorithm will strongly use the property that each $X_t$, $t\in V(T)$, induces a clique and in this section, we always assume that such a nice tree decomposition is given.
The following claim is the main ingredient that we will use in order to provide a polynomial algorithm.

\begin{claim}
  \label{claim:two}
  Let $G$ be a chordal graph,
  let $((T,\mathcal{X}), r, \mathcal{G})$ be a nice tree decomposition of $G$, and
  let $R$ be an induced subtree of $G$.
  For each $t \in V(T)$, $|V(R) \cap X_t| \leq 2$.
\end{claim}

\begin{proof}
  Assume that $|V(R) \cap X_t| \geq 3$, and let $u$, $v$, and $w$ be three distinct vertices of $V(R) \cap X_t$.
  As $\{u,v,w\} \subseteq X_t$, the three vertices induce a cycle.
  As $\{u,v,w\} \subseteq V(R)$, this contradicts the assertion that $R$ is an induced subtree.
  The claim follows.
\end{proof}

We now proceed to the description of a dynamic programming algorithm that processes the nice tree decomposition of the input graph.
Let $(G,v_0,\aaa,\bbb)$ be an instance of \textsc{Restricted Leafed Induced Subtree} such that $G$ is chordal and let $((T,\mathcal{X}), r, \mathcal{G})$ be a nice tree decomposition such that for each $t \in V(T)$, $X_t$ induces a clique in $G$.
Without loss of generality, up to a transformation that increases the size of $T$ by $O(n)$ and that can be done in linear time, we can assume that $((T,\mathcal{X}), r, \mathcal{G})$ is such that
$X_{r'} = \{v_0\}$ where $r'$ is the only child of $r$.

We define, for each $t \in V(T)$, the set
$\mathcal{I}_t = \{(S,d,i) \mid S \subseteq X_t, |S| \leq 2, d \colon S \to \{0,1,2\}, i\in [0,\aaa]\}$.
Using Claim~\ref{claim:two}, we know that if $|S| \geq 3$, then the associated element will not correspond to a partial solution.
This is why we can restrict our dynamic programming table to only consider the sets $S$ such that $|S| \leq 2$.
We also define a function $\dptable_t \colon \mathcal{I}_t \to \mathbb{N}$ such that for each $(S,d,i)\in \mathcal{I}_t$, $\dptable_t(S,d,i)$ is the maximum $\ell$ such that
there exists $\hs \subseteq V(G_t)$, called the \emph{witnesses} of $(S,d,i)$, that satisfies:
\begin{itemize}
\item $|\hs| = i$,
\item $\hs \cap X_t = S$,
\item if $v_0 \in V(G_t)$, then $v_0 \in \hat{S}$,
\item $G_t[\hs]$ is a tree,
\item for each $v \in S$, $d(v) = \min(\texttt{deg}_{G[\hs]}(v), 2)$, and
\item $\hs$ contains at least $\ell$ leaves in $V_t \setminus X_t$.
\end{itemize}

If such witnesses do not exist, we set $\dptable(S,d,i) = \bot$.

With this definition,
$(G,v_0,\aaa,\bbb)$ is a positive instance of \textsc{Restricted Leafed Induced Subtree} if and only if 
$\dptable_r(\emptyset,\varnothing,\aaa) \geq \bbb$.
For each $t \in V(T)$, we assume that we have already computed $\dptable_{t'}$ for each child $t'$ of $t$, and we proceed to the computation of $\dptable_t$.
We distinguish
several cases depending on the type of node $t$.

\begin{description}
\item[Leaf.] $\mathcal{I}_t = \{(\emptyset,\varnothing,i)\mid i \in [0,\aaa]\}$, $\dptable_t(\emptyset,\varnothing,0) = 0$ and for each $i \in [1,\aaa]$, $\dptable_t(\emptyset,\varnothing,i) = \bot$.
\item[Introduce vertex.] If $v$ is the insertion vertex of $X_t$ and $t'$ is the child of $t$, then for each $(S,d,i) \in \mathcal{I}_t$,

  \renewcommand{\arraystretch}{1.3}
  \hspace{-0.6cm}\begin{tabular}{l}
    $\dptable_t(S,d,i) =$ \\
    $ \max\big(\{\dptable_{t'}(S,d,i) \mid (S,d,i) \in \mathcal{I}_{t'}, v \not = v_0\}$ \\
    $\cup\ \{\dptable_{t'}(\emptyset,\varnothing,0) \mid (\emptyset,\varnothing,0) \in \mathcal{I}_{t'}, S = \{v\}, d = \{(v,0)\}, i = 1\}$ \\
    $\cup\ \{\dptable_{t'}(S',d',i-1) \mid (S',d',i-1) \in \mathcal{I}_{t'}, $\\
    $~~~~~~~~~~~~~~~~~~|S'| = 1,  S'= \{u\}, S = \{u,v\}, d(u) = \min(d'(u)+1,2), d(v) = 1\}).$\\
  \end{tabular}

\item[Forget vertex.] If  $v$ is the forget vertex of $X_t$ and $t'$ is the child of $t$, then for each $(S,d,i) \in \mathcal{I}_t$,

  \hspace{-0.6cm}\begin{tabular}{l}
      $\dptable_t(S,d,i) = $ \\
      $ \max\big(\{\dptable_{t'}(S,d,i) \mid (S,d,i) \in \mathcal{I}_{t'}, v \not \in S, v \neq v_0 \}$ \\
      $\cup\ \{\dptable_{t'}(S',d',i) +1 \mid (S',d',i) \in \mathcal{I}_{t'}, S' = S \cup \{v\}, |S| = 1, (v,1) \in d', d = d' \setminus \{(v,1)\}, v \neq v_0\})$ \\
      $\cup\ \{\dptable_{t'}(S',d',i) \mid (S',d',i) \in \mathcal{I}_{t'}, S' = S \cup \{v\}, |S| = 1,(v,2) \in d', d = d' \setminus \{(v,2)\}, v \neq v_0\}$ \\
      $\cup\ \{\dptable_{t'}(S',d',i) \mid (S',d',i) \in \mathcal{I}_{t'}, v = v_0, S' = \{v_0\}, S = \emptyset, (v_0,2) \in d', d = \varnothing\}.$ \\
    
  \end{tabular}


\item[Join.]  If $t'$ and $t''$ are the children of $t$, then for each $(S,d,i) \in \mathcal{I}_t$,

  \hspace{-0.6cm}\begin{tabular}{l}
    $\dptable_t(S,d,i) = $ \\
    $ \max(\{\dptable_{t'}(S,d',i') + \dptable_{t''}(S,d'',i'') \mid (S,d',i') \in \mathcal{I}_{t'}, (S,d'',i'')\in \mathcal{I}_{t''}, i = i' + i'' - |S|~~~~~~~~~~~~~~~~~~~~~~~~~$ \\
                    $ ~~~~~~~~~~~~~~~~~~~~~~~~~~~~~~~~~~~~~~~~~~~~~~~~~~~~~~~~~~~~~~\forall v \in S, d(v) = \min(2, d'(v) + d''(v) - (|S|-1))\}).$ \\
  \end{tabular}


\end{description}

We first discuss about the correctness of the algorithm.
When $t$ is a leaf, then by definition of the nice tree decomposition, $X_t$ is empty, thus the only $\hs$ that can satisfy the requirement is the empty set.
When $t$ introduces the vertex $v$, then either this vertex will not be in the partial solution, corresponding to the first line,  it is the first vertex to be added to the partial solution, corresponding to the second line, or it is a new vertex added to an already existing partial solution, corresponding to the third line.
If $t$ forgets the vertex $v$, then we distinguish the situation where $v$ is not in the partial solution we consider, first line,
it is a leaf of the induced subtree we are constructing, second line,
it is an internal vertex of the induced subtree we are constructing, third and fourth line depending whether the vertex is $v_0$ or not.
Indeed, note that, as we are working with the restricted version of \textsc{Leafed Induced Subtree}, we need to check that the vertex $v_0$ is an internal node.
In particular, when we forget a vertex, we need to ensure that this vertex will not be in a component disconnected to the component containing $v_0$ that may arrive later.
In the case where $t$ is a join node, we simply add the size of the partial solution of each side and check the new degree of each vertex in $S$.

Let us analyze the running time of this algorithm.
As, for each $t \in V(T)$, $S$  is a subset of $X_t$ of size at most $2$, $|X_t| \leq n$, $d$ can take at most $6$ values, and $i \in [0,\aaa]$, thus we have that $|\mathcal{I}_t| \leq 6 \cdot (\aaa +1) \cdot n^2$.
Note that
if $t$ is a leaf, then $\dptable_t$ can be computed in time $\Ocal(1)$,
if $t$ is an introduce vertex or a forget vertex node, and $t'$ is the child of $t$, then $\dptable_t$ can be computed in time $\Ocal(|\mathcal{I}_{t'}|\cdot |X_t|)$,
and if $t$ is a join node, and $t'$ and $t''$ are the two children of $t$, then  $\dptable_t$ can be computed in time $\Ocal(|\mathcal{I}_{t'}|\cdot |\mathcal{I}_{t''}| \cdot |X_t|)$.
Using the fact that, for a join node, we only consider the situation where the set $S$ is the same in both $\mathcal{I}_{t'}$ and $\mathcal{I}_{t''}$, we obtain that $\dptable_t$ can be computed in time $\Ocal(36\cdot a^2 \cdot n^2 \cdot |X_t|)$.
As the number of nodes of $T$ is linear, and the size of each bag is bounded by the number of vertices of the input graph, we obtain the following theorem.

\begin{theorem}
  \textsc{Restricted Leafed Induced Subtree} on an instance $(G,v_0,\aaa,\bbb)$ where $G$ is a chordal graph can be solved in time $\Ocal(\aaa^2\cdot n^4)$, where $n$ is the number of vertices of $G$.
\end{theorem}

\section{A single exponential algorithm parameterized by treewidth}
\label{sec:complexe}

In this section, we provide  an \textsf{FPT} algorithm parameterized by treewidth for the {\textsc{Restricted Leafed Induced Subtree}} problem.
We would like to mention that using classical dynamic programming techniques like done in Section~\ref{sec:simple}, we would obtain an algorithm running in time $2^{\Ocal(\tw \log \tw)}\cdot n^{\Ocal(1)}$.
The purpose of this section is to provide an algorithm running in time $2^{\Ocal(\tw)}\cdot \tw^{\Ocal(1)} \cdot n^7$.
For this, we use the techniques introduced by~\cite{BoCyKrNe2015}.
In order to describe the algorithm, we first need to restate some of the tools they introduced.

\bigskip

Let $U$ be a set.
We define $\Pi(U)$ to be the set of all partitions of $U$.
Given two partitions $p$ and $q$ of $U$, we define the coarsening relation $\sqsubseteq$ such that $p \sqsubseteq q$ if for each $S \in p$, there exists $S' \in q$ such that $S \subseteq S'$.
$(\Pi(U),\sqsubseteq)$ defines a lattice with maximum element $\{\{U\}\}$ and minimum element $\{\{x\}\mid x \in U\}$.
On this lattice, we denote by $\sqcap$ the meet operation and by $\sqcup$ the join operation.

Let $p \in \Pi(U)$.
For $X \subseteq U$ we denote by $p_{\downarrow X} = \{S \cap X \mid S \in p, S \cap X \not = \emptyset \}\in \Pi(X)$
 the partition obtained by removing all elements not in $X$ from $p$, and analogously for
$U \subseteq X$ we denote $p_{\uparrow X} = p \cup \{\{x\}\mid x \in X \setminus U\}\in \Pi(X)$  the partition obtained by adding to $p$ a singleton for each element in
$ X \setminus U$.
Given a subset $S$ of $U$, we define the partition $U[S] = \{\{x\}\mid x \in U\setminus S\} \cup \{S\}$.

A set of \emph{weighted partitions} is a set $\mathcal{A} \subseteq \Pi(U) \times \mathbb{N}$.
We also define $\rmc(\mathcal{A}) = \{(p,w) \in \mathcal{A} \mid \nexists (p,w') \in \mathcal{A}: w' < w \}$.

We now stats some operations on weighted partitions as defined by~\cite{BoCyKrNe2015}.
Let $U$ be a set and $\mathcal{A} \subseteq \Pi(U) \times \mathbb{N}$.

\begin{description}
\item[Union.] Given $\mathcal{B} \subseteq \Pi(U) \times \mathbb{N}$, we define $\mathcal{A} \cuparrow \mathcal{B} = \rmc(\mathcal{A} \cup \mathcal{B})$.
\item[Insert.] Given a set $X$ such that $X \cap U = \emptyset$, we define $\ins(X,\mathcal{A})= \{(p_{\uparrow U \cup X},w) \mid (p,w) \in \mathcal{A}\}$.
\item[Shift.] Given $w' \in \mathbb{N}$, we define $\shift(w',\mathcal{A}) = \{(p,w+w') \mid (p,w) \in \mathcal{A}\}$.
\item[Glue.] Given two elements $u$ and $v$, we define $\hat{U} = U \cup \{u,v\}$ and $\glue(\{u,v\},\mathcal{A}) \subseteq \Pi(\hat{U}) \times \mathbb{N}$ as\\$\glue(\{u,v\}, \mathcal{A}) = \rmc(\{(\hat{U}[uv] \sqcup p_{\uparrow \hat{U}}, w)\mid (p,w) \in \mathcal{A}\})$.
\item[Project.] Given $X \subseteq U$, we define $\overline{X} = U \setminus X$ and $\proj(X, \mathcal{A}) \subseteq \Pi(\overline{X}) \times \mathbb{N}$ as\\
 $\proj(X,\mathcal{A}) = \rmc(\{(p_{\downarrow \overline{X}},w) \mid (p,w) \in \mathcal{A}, \forall e \in X : \exists e' \in \overline{X} : p \sqsubseteq U[ee']\})$.
\item[Join.] Given a set $U'$, $\mathcal{B} \subseteq \Pi(U') \times \mathbb{N}$, and $\hat{U} = U \cup U'$, we define $\join(\mathcal{A}, \mathcal{B}) \subseteq \Pi(\hat{U}) \times \mathbb{N}$ as\\
 $\join(\mathcal{A}, \mathcal{B}) = \rmc(\{(p_{\uparrow \hat{U}}\sqcup q_{\uparrow \hat{U}}, w_1+w_2) \mid (p,w_1) \in \mathcal{A}, (q,w_2) \in \mathcal{B}\})$.
\end{description}

For the convenience of this paper, we also introduce the following operation, not present by \cite{BoCyKrNe2015}, that is a variation of glue.
\begin{description}
\item[Glue set.] Given a set $S = \{v_1, \ldots v_z\}$, we define $\hat{U} = U \cup S$ and $\glues(S,\mathcal{A}) \subseteq \Pi(\hat{U}) \times \mathbb{N}$ as\\
$\glues(S, \mathcal{A}) = \glue(\{v_1,v_2\}, \glue(\{v_1,v_3\}, \ldots, \glue(\{v_1,v_z\},\mathcal{A})\ldots)$.

\end{description}

\begin{proposition}[\cite{BoCyKrNe2015}]
  \label{prop:op}
  Each of the operations union, insert, shift, glue, and project can be carried out in time $s\cdot  |U|^{\Ocal(1)}$,  where $s$ is the size of the input of the operation.
Given two weighted partitions $\mathcal{A}$ and $\mathcal{B}$, $\join(\mathcal{A}, \mathcal{B})$ can be computed in time $|\mathcal{A}| \cdot |\mathcal{B}|\cdot |U|^{\Ocal(1)}$.
\end{proposition}

\begin{corollary}
  \label{coro:op}
  The glue set operation can be carried out in time $z \cdot s\cdot (|U| + z)^{\Ocal(1)}$,  where $s$ is the size of the input of the operation and $z$ is the size of the input set.
\end{corollary}
\begin{proof}
First notice that given a set $U$, $\mathcal{A} \subseteq \Pi(U) \times \mathbb{N}$, and two elements $u$ and $v$, $|\glue(\{u,v\},\mathcal{A})| \leq |\mathcal{A}|$. Indeed, $\{(\hat{U}[uv] \sqcap p_{\uparrow \hat{U}}, w)\mid (p,w) \in \mathcal{A}\}$ contains at most one element for each element of $\mathcal{A}$ and $\rmc$ produces a subset of its input set. This implies that at each call of $\glue$ in $\glues$, the size of the input of the operation does not increase, except for $U$. Thus, as there are $z-1$ calls to $\glue$, the operation glue set can be carried out in time $z \cdot s\cdot (|U| + z)^{\Ocal(1)}$.
\end{proof}

Given a weighted partition $\mathcal{A} \subseteq \Pi(U) \times \mathbb{N}$ and a partition $q \in \Pi(U)$, we define
$\opt(q,\mathcal{A}) = \min \{w \mid (p,w) \in \mathcal{A}, p \sqcup q = \{U\}\}$.
Given two weighted partitions $\mathcal{A},\mathcal{A}' \subseteq \Pi(U) \times \mathbb{N}$, we say that $\mathcal{A}$ \emph{represents} $\mathcal{A}'$ if for each $q \in \Pi(U)$, $\opt(q,\mathcal{A}) = \opt(q,\mathcal{A}')$.

Given a set $Z$ and a function $f \colon 2^{\Pi(U) \times \mathbb{N}} \times Z \to 2^{\Pi(U) \times \mathbb{N}}$, we say that $f$ \emph{preserves representation} if for each two weighted partitions $\mathcal{A},\mathcal{A}' \subseteq \Pi(U) \times \mathbb{N}$ and each $z \in Z$, it holds that if $\mathcal{A}'$ represents $\mathcal{A}$  then $f(\mathcal{A}', z)$ represents $f(\mathcal{A},z)$.

\begin{proposition}[\cite{BoCyKrNe2015}]
  \label{lemma:repr}
  The union, insert, shift, glue, project, and join operations preserve representation.
\end{proposition}

As the glue set operation is only a succession of glue operations, we also have that the glue set operation preserves representation.

\begin{corollary}
  \label{coro:repr}
  The glue set operation preserves representation.
\end{corollary}

\begin{theorem}[\cite{BoCyKrNe2015}]
  \label{th:reduce}
  There exists an algorithm $\reduce$ that, given a set of weighted partitions $\mathcal{A} \subseteq \Pi(U) \times \mathbb{N}$, outputs in time
  $|\mathcal{A}|\cdot 2^{(\omega-1) |U|}\cdot |U|^{\Ocal(1)}$ a set of weighted partitions $\mathcal{A}'\subseteq \mathcal{A}$ such that $\mathcal{A}'$ represents $\mathcal{A}$ and $|\mathcal{A}'| \leq 2^{|U|}$, where $\omega$ denotes the matrix multiplication exponent.
\end{theorem}

In the remainder of this section, the weight of the weighted partitions will be irrelevant.
In particular, we can assume that all our weighted partitions will be subsets of $\Pi(U) \times \{0\}$.
Because of this, we will abuse the notation and consider the weighted partitions as subsets of $\Pi(U)$, keeping in mind that the weight associated to each partition exists but is irrelevant.

Inspired by the algorithm provided by~\cite{BoCyKrNe2015} in their Section 3.5, we present an algorithm solving {\textsc{Restricted Leafed Induced Subtree}} on an instance $(G,v_0,a,b)$.
Given a nice tree decomposition of $G$ of width $w$, we define a nice tree decomposition  $((T,\Xcal), r, \Gcal)$ of $G$ of width at most $w+1$
such that the only empty bags are the root and the leaves and for each $t \in T$, if $X_t \not = \emptyset$ then $v_0 \in X_t$.
Note that this can be done in linear time.
For each bag $t$,  each integer $i$, $j$, and $\ell$,
each function $\sbf \colon X_t \rightarrow \{0,1,2_0,2_1\}$,
and each partition $p \in \Pi(\sbf^{-1}(1))$,
we define:

\begin{eqnarray*}
  \Ecal_t(p,\sbf, i,j,\ell) & = &
                                  \{(I,L) \mid (I,L) \in 2^{V_t} \times  2^{V_t},\ I \cap L = \emptyset,\\
                            &&~~~ |I|  = i,\ |L| = \ell, |E(G_t[I])| = j,\\
                            && ~~~ \forall v \in L \cap X_t,\ \sbf(v) = 2_z\mbox{ with }z = {\degs{G_t[I \cup L]}{v}},\\
                            && ~~~ \forall v \in L \setminus X_t,\ {\degs{G_t[I \cup L]}{v} = 1},\\
                            &&~~~  \forall u,v \in L, \{u,v\} \not \in E(G_t),\\
                                &&~~~ I \cap X_t = \sbf^{-1}(1),\ v_0 \in X_t \Rightarrow \sbf(v_0) = 1, \\ 
                            &&~~~ \forall u \in (I \cup L)\setminus X_t : \mbox{ either $t$ is the root and $I$ is connected or } \\
                            &&~~~~~~~~~~~~~\exists v \in \sbf^{-1}(1): \mbox{$u$ and $v$ are connected in $G_t[I \cup L]$,}\\
                            &&~~~ \forall u,v \in \sbf^{-1}(1): \sbf^{-1}(1)[\{u,v\}] \sqsubseteq p \Leftrightarrow  \\ 
                            &&~~~~~~~~~~~~~~~~~~~~~~~~~~~~~~~~~\mbox{$u$ and $v$ are connected in $G_t[I \cup L]$} \},\\
  \\
  \Acal_t(\sbf, i,j,\ell) & = &
                                       \{p \mid p \in \Pi(\sbf^{-1}(1)),\ \Ecal_t(p,\sbf,i,j,\ell) \not = \emptyset\}.\\
\end{eqnarray*}

In the definition of $\Ecal_t$, the set $I$ (resp. $L$) corresponds to the internal vertices (resp. the leaves) of the required induced subtree when restricted to $G_t$.
In the algorithm, we ensure that $I$ will induce a tree and that each vertex in $L$ will be connected to $I$ by exactly one edge.
In order to verify that the vertices in $I$ induce a tree, we will simply check that $G[I]$ is connected and has exactly $|I|-1$ edges. For the leaf, during the algorithm, we divide them into two groups, those that are labeled $2_0$ that have no neighbors in $I$ yet, and those labeled $2_1$ that already has exactly one neighbor in $I$.
Note that as we want to maximize the number of leaves, we do not care whether the vertices in $I$ are leaves or not.
Thus, a given instance $(G,v_0,a,b)$ of {\textsc{Restricted Leafed Induced Subtree}} is a positive instance if and only if for some  $\ell \geq b$ and $i = a - \ell$, we have $\Acal_r(\ef,i,i-1,\ell) \not = \emptyset$.

For each $t \in V(T)$, we assume that we have already computed $\Acal_{t'}$ for every child $t'$ of $t$, and we proceed to the computation of $\Acal_t$.
We distinguish
several cases depending on the type of node $t$.

\begin{description}
\item[Leaf.] By definition of $\Acal_t$, we have $\Acal_t(\ef,0,0,0) = \{\emptyset\}$.

\item[Introduce vertex.]
  Let $v$ be the insertion vertex of $X_t$, let $t'$ be the child of $t$, let $\sbf \colon X_t \rightarrow \{0,1,2_0,2_1\}$, and let $H = G_t[\sbf^{-1}(1)]$.
  \begin{itemize}
  \item If  $v = v_0$ and $\sbf(v_0) \in \{0,2_0,2_1\}$, then by definition of $\Acal_t$ we have that $\Acal_t(\sbf,i,j,\ell) = \emptyset$.
  \item Otherwise, if $v = v_0$, then by construction of the nice tree decomposition,
    we know that $t'$ is a leaf of $T$ and so
    $\sbf = \{(v_0,1)\}$, $j = \ell = i-1 = 0$ and $\Acal_t(\sbf, i,j,\ell) = \ins(\{v_0\}, \Acal_{t'}(\ef,0,0,0)$).
  \item Otherwise, if $\sbf(v) = 0$, then, by definition of $\Acal_t$, it holds that $\Acal_t(\sbf,i,j,\ell) = \Acal_{t'}(\sbf|_{X_{t'}},i,j,\ell)$.
    
  \item Otherwise, if $\sbf(v) = 2_z$, $z \in \{0,1\}$, then let $R=N_{G_t[X_t]}(v)\setminus \sbf^{-1}(0)$.
    If $|R| = z$ and $R \subseteq \sbf^{-1}(1)$ then $\Acal_t(\sbf, i,j,\ell) = \Acal_{t'}(\sbf|_{X_{t'}},i,j,\ell-1)$.
    Otherwise $\Acal_t(\sbf, i,j,\ell) = \emptyset$.
  \item
    Otherwise, we know that $v \not = v_0$ and $\sbf(v) = 1$.
    If $N_{G_t[X_t]}(v) \cap \sbf^{-1}(2_0) \not = \emptyset$, then $\Acal_t(\sbf, i,j,\ell) = \emptyset$.
    Otherwise,
    with $L_1 = N_{G_t[X_t]}(v) \cap  \sbf^{-1}(2_1)$, 
    let $\sbf' \colon X_{t'} \to \{0,1,2_0,2_1\}$ defined such that $\forall v' \in X_{t'} \setminus L_1$, $\sbf'(v') = \sbf(v')$ and for each $v' \in L_1$, $\sbf'(v') = 2_{0}$.
    Then 
  \begin{eqnarray*}
      \Acal_t(\sbf,i,j,\ell) & =  \glues(N_{H}[v],\ins(\{v\}, \Acal_{t'}(\sbf',i-1,j - |N_H(v)|,\ell))),
  \end{eqnarray*}
  where $N_{H}(v)$ denotes the open neighborhood of $v$ in $H$, i.e., all the neighbors $u$ of $v$ in $G_t[X_t]$ such that $s(u) = 1$, and $N_{H}[v]$ denotes the close neighborhood of $v$ in $H$, i.e., the set $N_{H}(v) \cup \{v\}$.
  \end{itemize}

\item[Forget vertex.]
  Let $v$ be the forget vertex of $X_t$, let $t'$ be the child of $t$, and let  $\sbf \colon X_t \rightarrow \{0,1,2_0,2_1\}$.
  As a vertex from the leaf part can be removed only if it has exactly one neighbor and all other vertices can be removed safely, we obtain that
  \begin{eqnarray*}
    \Acal_t(\sbf,i,j,\ell) &=&
                               \mathcal{A}_{t'}(\sbf\cup \{(v,0)\},i,j,\ell)\\
                           &&\cuparrow \proj(\{v\},\mathcal{A}_{t'}(\sbf \cup \{(v,1)\},i,j,\ell))\\
                           && \cuparrow \mathcal{A}_{t'}(\sbf\cup \{(v,2_1)\},i,j,\ell).\\
  \end{eqnarray*}
\item[Join.]
  Let $t'$ and $t''$ be the two children of $t$, let $\sbf \colon X_t \rightarrow \{0,1,2_0,2_1\}$, and let $H =G_t[\sbf^{-1}(1)]$.
%
%
  Given three functions $\sbf^*,\sbf', \sbf'' \colon X_t \rightarrow \{0,1,2_0,2_1\}$, we say that $\sbf^* = \sbf' \oplus_t \sbf''$ if for each $v \in \sbf^{*-1}(\{0,1\})$, $\sbf^*(v) = \sbf'(v) = \sbf''(v)$, and for each $v \in X_t$ such that $\sbf^*(v) = 2_z$, $z \in \{0,1\}$, we either have:
  \begin{itemize}
  \item $N_{G_t}(v) \cap \sbf^{*-1}(1) = \emptyset$, $s'(v)=2_{z'}$, $s''(v)=2_{z''}$, and $z = z' + z''$, or
  \item $|N_{G_t}(v) \cap \sbf^{*-1}(1)| = 1$ and $\sbf^*(v) = \sbf'(v) = \sbf''(v) = 2_1$.
  \end{itemize}
  Two entries $A_{t'}(\sbf',i',j',\ell')$ and $A_{t''}(\sbf'',i'',j'',\ell'')$ are compatible if  $\sbf' \oplus_t \sbf''$ is defined.
  We join every pair of compatible entries
  $A_{t'}(\sbf',i',j',\ell')$ and $A_{t''}(\sbf'',i'',j'',\ell'')$.
  In this case, notice that the internal vertices remain internal vertices, and the same for the leaves.
  We only have to check the size of the neighborhood of each leaf depending on whether the neighbors are within $X_t$ or not.
  We obtain that
  \begin{eqnarray*}
    \Acal_t(\sbf,i,j,\ell) &=&\!\!\!\!\!\!\!\!\!\!
                                      \underset{\ell' + \ell'' = \ell + | \sbf^{*-1}(\{2_0,2_1\})| }{
                                      \underset{j'+j'' = j + |E(H)|}{
                                      \underset{i'+i''=i+|V(H)|}{
                                      \underset{\sbf = \sbf' \oplus_t \sbf''}{
                                      {
                                      \underset{\sbf', \sbf'' \colon X_t \rightarrow \{0,1,2_0,2_1\},}{
                                      \underset{}
                                      \bigcuparrow
                                      }}}}}}\!\!\!\!\!\!\!
                                      \join(\mathcal{A}_{t'}(\sbf',i',j',\ell'), \mathcal{A}_{t''}(\sbf'', \rbf'',i'',j'',\ell'')).
  \end{eqnarray*}


\end{description}

\begin{theorem}
  {\textsc{Restricted Leafed Induced Subtree}} can be solved in time $2^{\Ocal(\tw)}\cdot \tw^{\Ocal(1)} \cdot  n^7$.
\end{theorem}

\begin{proof}
  As discussed in Section~\ref{sec:prelim}, we can assume that a nice tree decomposition of width $\ww = 2 \cdot \tw + 1$
is given with the input.
  The algorithm works in the following way.
  For each node $t \in V(T)$ and for each entry $M$ of its table, instead of storing $\Acal_t(M)$, we store $\Acal'_t(M) = \reduce(\Acal_t(M))$ by using Theorem~\ref{th:reduce}.
  As each of the operations we use preserves representation by Proposition~\ref{lemma:repr} and Corollary~\ref{coro:repr}, we obtain that for each node $t \in V(T)$ and for each possible entry $M$, $\Acal'_t(M)$ represents $\Acal_t(M)$.
  In particular, we have that $\Acal'_r(M) = \reduce(\Acal_r(M))$ for each possible entry $M$.
  Using the definition of $\Acal_r$ we have that $(G,v_0,a,b)$ is a positive instance of 
  {\textsc{Restricted Leafed Induced Subtree}} 
  if and only if for some $i$ and $\ell$, $\ell \geq b$, $i+\ell = a$, and $\Acal'_r(\ef,i,i-1,\ell) \not = \emptyset$.

  We now focus on the running time of the algorithm.
  For this we proceed to explain, for each $t \in V(T)$, the running time needed for producing $\mathcal{A}'_t(M)$, for each entry $M$, assuming that for each child $t'$ of $t$ in $T$, the size of $\mathcal{A}'_{t'}(M')$ is of size at most $2^{\ww}$ for each entry $M'$. We distinguish the cases depending on the type of node $t$.
\begin{description}
\item[Leaf.] It is clear that it can be done in time $\Ocal(1)$.
\item[Introduce vertex.] The new set $\mathcal{A}_t$ is constructed using either an insert operation, a glue set operation with a set of size at most $\ww$, or no operation at all. By Proposition~\ref{prop:op} and Corollary~\ref{coro:op}, applying the operations take at most $2^{\ww} \cdot \ww^{\Ocal(1)}$ steps and, by Theorem~\ref{th:reduce}, the reduce algorithm takes at most $2^{\Ocal(\ww)} \cdot \ww^{\Ocal(1)}$ steps.
  As this need to be performed for each of the at most $4^\ww\cdot n^3$ possible entries $M$ of $\mathcal{A}_t$, the introduce vertex  node can be computed in time $2^{\Ocal(\ww)} \cdot \ww^{\Ocal(1)} \cdot n^3$.
\item[Forget vertex.] We take the union of three sets of size at most $2^{\ww}$, each of them can be computed, by Proposition~\ref{prop:op}, in time $2^{\ww} \cdot \ww^{\Ocal(1)}$ and, by Theorem~\ref{th:reduce}, the reduce algorithm runs in time $2^{\Ocal(\ww)} \cdot \ww^{\Ocal(1)}$.
  As this need to be performed for each of the at most $4^\ww\cdot n^3$ possible entries $M$ of $\mathcal{A}_t$, the forget vertex  node can be computed in time $2^{\Ocal(\ww)} \cdot \ww^{\Ocal(1)} \cdot n^3$.
\item[Join.] Given an entry $M' = (\sbf',i',j',\ell')$ of $\mathcal{A}'_{t'}$ and an entry $M''= (\sbf'',i'',j'',\ell'')$ of $\mathcal{A}'_{t''}$, by Proposition~\ref{prop:op}, $\join(\mathcal{A}'_{t'}(M'),\mathcal{A}'_{t''}(M''))$ can be computed in time $2^{2\cdot \ww} \cdot \ww^{\Ocal(1)}$.
  Moreover, the number of times this join operation is used for computing $\mathcal{A}'_t(M)$ for a specific entry $M = (\sbf,i,j,\ell)$ is given by the fact that we consider at most $4^{\ww} $ possible functions $\sbf'$, as many functions $\sbf''$, at most $n+\ww$ choices for $i'$ and $i''$,
  at most $n+ \ww$ choices for $j'$ and $j''$ (as we can always assume, during the algorithm, that $H$ is a forest), and at most $n+\ww$ choices for $\ell'$ and $\ell''$. This implies that for this specific $M$, we execute $ (n+\ww)^3 \cdot 16^{\ww}$ times the join operation and then, by Theorem~\ref{th:reduce}, the reduce algorithm can be executed in time $2^{\Ocal(\ww)} \cdot \ww^{\Ocal(1)}\cdot n^3$. As we need to do it for each of the at most $4^\ww \cdot n^3$ possible entries $M$, the join node can be computed in time $2^{\Ocal(\ww)} \cdot \ww^{\Ocal(1)} \cdot n^6$.
\end{description}
  The theorem follows by taking into account the linear number of nodes in a nice tree decomposition.
\end{proof}

\section{Conclusion}

In this paper we close the open question asked by~\cite{BlCaGoLaNaVa2018}
by providing a polynomial algorithm for  {\textsc{Restricted Leafed Induced Subtree}} when the input graph is chordal as well as an \textsf{FPT} algorithm parameterized by treewidth.

This paper provides a generalization of the algorithms provided by~\cite{BlCaGoLaNaVa2018} and~\cite{AbBlGo2018} and the natural question would be whether it can be even more generalized or if there are other graph classes for which the problem is polynomial.

Another more ambitious direction would be to consider the \textsc{Fully Leafed Induced Subtrees} problem as a specific instance of a class of \textsc{Induced Subtrees} problems and provide a generic algorithm that solves problems of this class on chordal graphs or on graphs of bounded treewidth in a similar way that the one  presented by~\cite{fvBaWa2022} for non-induced subtrees.

\acknowledgements
\label{sec:ack}
The author thanks Élise Vandomme for helpful discussions that lead to the writing of this paper.

\nocite{*}
\bibliographystyle{abbrvnat}
\bibliography{biblio}
\label{sec:biblio}

\end{document}